\documentclass{article}
\usepackage[top=4cm, bottom=4cm, left=3.3cm, right=3.3cm, includefoot]{geometry}

\usepackage{amsmath,amsfonts,amssymb,amsthm,tikz}
\usepackage{bm}
\usepackage{xcolor}
\usepackage{multirow}
\usepackage{url}
\usepackage[all]{xy}
\usepackage{hyperref}
\usepackage[capitalize]{cleveref}
\usepackage{booktabs}
\usepackage{mathtools}
\mathtoolsset{showonlyrefs,showmanualtags}
\usepackage{physics}
\usepackage[linesnumbered,ruled,vlined]{algorithm2e}

\SetCommentSty{mycommfont}

\SetKwInput{KwInput}{Input}                %
\SetKwInput{KwOutput}{Output}              %

\newcommand{\Span}[1]{\mathrm{span}_{\mathbb{R}}({#1})}

\newcommand{\degree}[1]{\mathrm{deg}({#1})}
\newcommand{\degreek}[2]{\mathrm{deg}_{{#1}}({#2})}
\newcommand{\gnorm}[1]{\norm{{#1}}_{\mathrm{g},X}}
\newcommand{\cnorm}[1]{\norm{{#1}}_{\mathrm{c}}}
\newcommand{\gwnorm}[1]{\norm{{#1}}_{\mathrm{gw},X}}
\newcommand{\gwnormx}[2]{\norm{{#1}}_{\mathrm{gw},{#2}}}

\newcommand{\pdiv}[2]{\frac{\partial {#1}}{\partial {#2}}}

\theoremstyle{definition}
\newtheorem{theorem}{Theorem}[section]
\newtheorem{lemma}[theorem]{Lemma}

\newtheorem{proposition}[theorem]{Proposition}
\newtheorem{definition}[theorem]{Definition}

\newtheorem{example}[theorem]{Example}
\newtheorem{remark}[theorem]{Remark}

\usepackage[whole]{bxcjkjatype}
\usepackage[normalem]{ulem}
\usepackage{xcolor} 
\usepackage{comment}
\usepackage{soul}

\begin{document}

\title{Border basis computation \\ with gradient-weighted normalization}

\author{Hiroshi Kera}
\author{Hiroshi Kera\footnote{Graduate School of Engineering, Chiba University. Corresponding author: Hiroshi Kera (e-mail: kera.hiroshi@gmail.com).}}
\date{}

\maketitle

\begin{abstract}
    Normalization of polynomials plays a vital role in the approximate basis computation of vanishing ideals. Coefficient normalization, which normalizes a polynomial with its coefficient norm, is the most common method in computer algebra. This study proposes the gradient-weighted normalization method for the approximate border basis computation of vanishing ideals, inspired by recent developments in machine learning. The data-dependent nature of gradient-weighted normalization leads to better stability against perturbation and consistency in the scaling of input points, which cannot be attained by coefficient normalization.
    Only a subtle change is needed to introduce gradient normalization in the existing algorithms with coefficient normalization. 
    The analysis of algorithms still works with a small modification, and the order of magnitude of time complexity of algorithms remains unchanged.
    We also prove that, with coefficient normalization, which does not provide the scaling consistency property, scaling of points (e.g., as a preprocessing) can cause an approximate basis computation to fail. This study is the first to theoretically highlight the crucial effect of scaling in approximate basis computation and presents the utility of data-dependent normalization.
\end{abstract}

\maketitle

\section{Introduction}
Given a set of points $X \subset \mathbb{R}^n$,  the vanishing ideal of $X$ is the set of polynomials in $\mathbb{R}[x_1,\ldots,x_n]$ that vanish for any $\mathbf{x}\in X$. 
\begin{align}
\mathcal{I}(X) = \qty{g \in \mathbb{R}[x_1,\ldots,x_n] \mid \forall \mathbf{x} \in X, g(\mathbf{x}) = 0}. 
\end{align}
The approximate computation of bases of vanishing ideals has been extensively studied~\cite{abbott2008stable,heldt2009approximate,fassino2010almost,robbiano2010approximate,limbeck2013computation,livni2013vanishing,kiraly2014dual,kera2018approximate,kera2019spurious,kera2020gradient,wirth2022conditional} in the last decade, where a basis comprises approximately vanishing polynomials, i.e., $g(\mathbf{x})\approx 0, (\forall\mathbf{x}\in X)$. Approximate basis computation and approximately vanishing polynomials are exploited in various fields such as dynamics reconstruction, signal processing, and machine learning~\cite{torrente2009application,hou2016discriminative,kera2016vanishing,kera2016noise,wang2018nonlinear,wang2019polynomial,antonova2020analytic,karimov2020algebraic}. A wide variety of applications is possible because the approximate basis computation takes a set of noisy points as its input---suitable for the recent data-driven applications---and efficiently computes a set of multivariate polynomials that characterize the given data.

Coefficient normalization, where polynomials are normalized to gain a unit coefficient norm, is the most common choice in computer algebra. In contrast, in machine learning, the basis computation of vanishing ideals is performed in a monomial-agnostic manner to sidestep symbolic computation and term orderings~\cite{livni2013vanishing,kiraly2014dual,kera2019spurious}. In this case, efficient access to the coefficients of terms is not possible. Thus, polynomials are handled without proper normalization. A recent study solved this problem by gradient normalization~\cite{kera2020gradient}, which used the gradient (semi-)norm $\sqrt{\sum_{x \in X} \norm{\nabla g(x)}^2}$. Interestingly, the data-dependent nature of gradient normalization provides new properties that could not be realized by other basis computation algorithms. However, the direct application of gradient normalization to the monomial-aware basis computation in computer algebra does not take over these advantages and merely increases the computational cost. Thus, an effective data-dependent normalization remains unexplored for computer-algebraic approaches. 

In this study, we propose a new normalization, called gradient-weighted normalization, which is a hybrid of coefficient normalization conventionally used in computer algebra and the gradient normalization recently developed in machine learning. Gradient-weighted normalization can be applied to several existing basis computation algorithms for vanishing ideals in computer algebra. In particular, we focus on the approximate computation of border bases because these are the most common choices in the approximate computation~\cite{abbott2008stable,heldt2009approximate,limbeck2013computation} as they have greater numerical stability than the Gr\"obner bases~\cite{stetter2004numerical,fassino2010almost}. 
We highlight the following advantages of gradient-weighted normalization in the approximate border basis computation. As an example, we analyze the approximate Buchberger--M\"oller (ABM) algorithm~\cite{limbeck2013computation}. 
\begin{itemize}
\item Gradient-weighted normalization realizes an approximate border basis computation that outputs polynomials that are more robust against perturbation on the input points. 
\item With gradient-weighted normalization, eigendecomposition-based (or singular value decomposition (SVD)-based) approximate border basis computation methods are equipped with scaling consistency; scaling input points does not change the configuration of the output basis and only linearly scales the evaluation values for the input points.
\item Gradient-weighted normalization only requires a small modification to an algorithm to work with, causing subtle changes in the analysis of the original one. Unlike gradient normalization, gradient-weighted normalization does not change the order of magnitude of time complexity.
\end{itemize}

In particular, the second advantage, the scaling consistency, provides us an important insight into approximate basis computation: without it, not only the approximation tolerance $\epsilon$ but also a scaling factor $\alpha$ of points must be properly chosen. In Proposition~\ref{prop:coefficient-normalization-failure}, we prove that, under a mild condition, an approximate basis computation with coefficient normalization always fails if the scaling factor is not properly set. This result implies that even preprocessing of points (e.g., scaling points to range in $[-1, 1]$) for numerical stability can cause a failure of the approximate basis computation. 
We consider that this study reveals a new direction in approximate border basis computation toward data-dependent normalization and its analysis.

\section{Related Work}
The gradient of polynomials has been exploited for approximate computation of vanishing ideals in several studies. In~\cite{abbott2008stable}, the first-order approximation (and thus gradient) of polynomials was computed to discover a set of monomials with their evaluation matrix still in full-rank for small perturbations in the points. Similarly, ~\cite{fassino2013simple} considered the first-order approximation of polynomials to compute a low-degree polynomial that approximately passed through the given points in terms of the geometrical distance. However, the former incurred a heavy computational cost and strong sensitivity to the hyperparameter $\gamma$, while the latter only focused on the lowest degree polynomial and did not give a basis. Furthermore, both methods used coefficient normalization.

In ~\cite{kera2020gradient}, which is the most similar to this study, polynomials normalized by the gradient norm were considered, which is, to the best our knowledge, the first data-dependent normalization to compute approximately vanishing polynomials. However, their method focuses on monomial-agnostic basis computation, where coefficients of terms are inaccessible. Although this can help when symbolic computation and term ordering are unfavorable, how helpful data-dependent normalization is in the monomial-\textit{aware} setting---the standard in computer algebra---is still unknown. The direct application of gradient normalization in border basis computation cannot fully exploit the advantages of monomial-agnostic basis computation. Furthermore, the method to relate the gradient norm to the coefficient norm, which plays an important role in approximate border bases, remains unclear. In this study, we propose gradient-weighted normalization, which brings all the merits of gradient normalization into monomial-aware computation while retaining the same order of magnitude of the time complexity of algorithms.
Furthermore, by exploiting the monomial-aware setting, a more detailed analysis is performed with gradient-weighted normalization. Particularly, the gradient-weighted norm of the terms and polynomials in basis computation can be lower and upper bounded. In addition, the coefficient norm can be upper bounded by the gradient-weighted norm.

\section{Preliminaries}
We consider a finite set of points $X\subset\mathbb{R}^n$, a polynomial ring $\mathcal{R}_n = \mathbb{R}[x_1,\ldots, x_n]$, and set of terms $\mathcal{T}_n \subset \mathcal{R}_n$, where $x_1,\ldots, x_n$ are indeterminates, throughout the paper. The vanishing ideal $\mathcal{I}(X)\subset\mathcal{R}_n$ is thus zero-dimensional. The definitions of the order ideal and the border basis are based on those in~\cite{kreuzer2005computational,kehrein2005charactorizations}, while the definitions of approximate notions are based on~\cite{heldt2009approximate}.

\begin{definition}
Given a set of points $X = \{\mathbf{x}_1,\mathbf{x}_2,...,\mathbf{x}_{N}\} \subset \mathbb{R}^n$, with gentle abuse of notation, the \textbf{evaluation vector} of a polynomial $h\in\mathcal{R}_n$ and its gradient $\nabla h$ are defined
as follows, respectively. 
\begin{align*}
h(X) & =\mqty(h(\mathbf{x}_{1}) & h(\mathbf{x}_{2}) & \cdots & h(\mathbf{x}_{N}))^{\top}\in\mathbb{R}^{N},\\
\nabla h(X) & = \mqty(\nabla h(\mathbf{x}_{1})^{\top} & \nabla h(\mathbf{x}_{2})^{\top} & \cdots & \nabla h(\mathbf{x}_{N})^{\top})^{\top}\in\mathbb{R}^{nN}.
\end{align*}
For a set of polynomials $H=\qty{ h_{1},h_{2},\ldots,h_{s}} \subset\mathcal{R}_n$ and the set of their gradients $\nabla H=\qty{ \nabla h_{1},\nabla h_{2},\ldots,\nabla h_{s}}$,
each \textbf{evaluation matrix} is defined as 
\begin{align*}
	H(X) & = \mqty(h_{1}(X) & h_{2}(X) & \cdots & h_{s}(X))\in\mathbb{R}^{N\times s}, \\
	\nabla H(X) & = \mqty(\nabla h_{1}(X) & \nabla h_{2}(X) & \cdots  & \nabla h_{s}(X))\in\mathbb{R}^{nN\times s}. 
\end{align*}
\end{definition}

\begin{definition}
A polynomial $f \in \mathcal{R}_n$ is said to be \textbf{unitary} if the norm of its coefficient vector equals one.
\end{definition}

\begin{definition}
A finite set of terms $\mathcal{O}\subset \mathcal{T}_n$ is called an \textbf{order ideal} if the following holds: if $t\in \mathcal{T}_n$ divides $o\in\mathcal{O}$, then $t \in \mathcal{O}$. The \textbf{border} of $\mathcal{O}$ is defined as $\partial\mathcal{O} = \qty(\bigcup_{k=1}^n x_k\mathcal{O}) \backslash \mathcal{O}$. 
\end{definition}

\begin{definition}\label{def:border-basis}
Let $\mathcal{O} \subset \mathcal{T}_n$ be an order ideal. Then, an \textbf{$\mathcal{O}$-border prebasis} $G$ is a set of polynomials in the form $b - \sum_{o\in O} c_o o$,
where $b \in \partial \mathcal{O}$, and $c_o \in \mathbb{R}$. If $\mathcal{O}$ is a basis of the $\mathbb{R}$-vector space $\mathbb{R}[x_1,x_2,\ldots, x_n]/I$, then $G$ is called an \textbf{$\mathcal{O}$-border basis} of an ideal $I$. 
\end{definition}

\begin{definition}
Given $\epsilon \ge 0$, a polynomial $g\in\mathcal{R}_n$ is said to be an \textbf{$\epsilon$-approximately vanishing} for a set of points $X\subset \mathbb{R}^n$, if $\|g(X)\|\le\epsilon$, where $\|\cdot\|$ denotes the $L_2$ norm.
\end{definition}

\begin{definition}
Given $\epsilon \ge 0$, an ideal $\mathcal{I}\subset\mathcal{R}_n$ is said to be an \textbf{$\epsilon$-approximate vanishing ideal} for a set of points $X\subset \mathbb{R}^n$ if there exists a system of unitary polynomials that generates $\mathcal{I}$ and is $\epsilon$-approximately vanishing for $X$. 
\end{definition}

\begin{remark}\label{rem:O-basis-property}
Let us consider the $\mathcal{O}$-border basis $G$ of the vanishing ideal $\mathcal{I}(X)$ of $X\subset\mathbb{R}^n$. The evaluation vectors of the order terms span $\mathbb{R}^{|X|}$. The evaluation vectors of the terms in $\mathcal{O}$ are linearly independent, and $|\mathcal{O}| = |X|$, where $|\cdot|$ denotes the cardinality of set. In the approximate case, the former still holds, and the latter becomes $|\mathcal{O}| \le |X|$.
\end{remark}

\paragraph{Other notation} We denote the support of a given polynomial by $\mathrm{supp}(\cdot)$ and the set of linear combinations of a given set of terms with coefficients in $\mathbb{R}$ by $\Span{\cdot}$. Further, $\cnorm{\cdot}$ denotes the coefficient norm of a polynomial (i.e., $\cnorm{g}=\sqrt{\sum_{i} c_i^2}$ for $g=\sum_{i}c_it_i$, where $c_i\in\mathbb{R}$ and $t_i\in\mathcal{T}_n$). The total degree of a polynomial is denoted by $\degree{\cdot}$ and $\degreek{k}{\cdot}$ denotes the degree of polynomial with respect to $x_k$. 

\section{Border bases with gradient-weighted normalization}
\begin{definition}\label{def:gradient-norm}
The \textbf{gradient norm} of a polynomial $g\in\mathcal{R}_n$ with respect to $X\subset\mathbb{R}^n$ is
$\gnorm{g} = \sqrt{\sum_{x \in X} \norm{\nabla g(x)}^2}/Z$,
where $Z = \sqrt{\sum_{k=1}^n \mathrm{deg}_k(g)^2}$ and $\gnorm{g}:=0$ if $g$ is a constant polynomial.
\end{definition} 

\begin{definition}
The \textbf{gradient-weighted norm}\footnote{Strictly speaking, this is a semi-norm because $\gwnorm{f}=0$ (and $\gnorm{f}=0$) for $f \in \mathcal{R}_n$ does not imply $f=0$. However, all the terms (except 1) and polynomials appearing in the border basis computation do not vanish with respect to the gradient-weighted norm. 
For simplicity, we refer to $\gwnorm{\cdot}$ (and $\gnorm{\cdot}$) as a norm in this study.} of a polynomial $g = \sum_{i}c_i t_i, (c_i\in\mathbb{R},t_i\in\mathcal{T}_n)$ is defined by $\gwnorm{g} = \sqrt{\sum_{i} c_i^2 \|t_i\|_{g,X}^2}$.
If the gradient-weighted norm of $g$ is equal to one, then $g$ is \textbf{gradient-weighted unitary}.
\end{definition}

\begin{remark}
For any term $t\in\mathcal{T}_n$, its gradient norm and gradient-weighted norm are identical, i.e., $\gnorm{t} = \gwnorm{t}$. 
The proofs in this study work with any constant $Z>0$; however, our choice of $Z$ provides simpler bounds.
\end{remark}

In general, the gradient-weighted norm and coefficient norms of a polynomial are not always correlated; a large gradient-weighted norm does not necessarily imply a large coefficient norm and vice versa. The following two examples illustrate this:

\begin{example}\label{example:zero-gcnorm}
Let us consider a polynomial $f = x^2y^2 - c \in \mathbb{R}[x,y], (c\in\mathbb{R})$. The gradient-weighted norm of $f$ is $\gwnorm{f} = 0$ for $X = \{(1, 0), (0, 1)\}$, whereas the coefficient norm $\cnorm{f}=\sqrt{1+c^2}$ can be arbitrarily enlarged by increasing $|c|$. 
\end{example}

\begin{example}\label{example:large-gcnorm}
Let us consider a polynomial $f = (x^2 + y^2 - 1)/\sqrt{3} \in \mathbb{R}[x,y]$. The coefficient norm of $f$ is $\cnorm{f}=1$, whereas the gradient-weighted norm for $X = \{(k, 0), (0, k)\}$ is $\gwnorm{f}=2\abs{k}/\sqrt{3}$, which 
can be arbitrarily enlarged by increasing $|k|$. 
\end{example}

Example~\ref{example:zero-gcnorm} also indicates that normalizing polynomials with their gradient-weighted norms is not always a valid approach because it could lead to zero-division. However, we can prove that gradient-weighted normalization is always valid in border basis computation. First, we prove the following lemma. 
\begin{lemma}\label{lemma:odiv-is-o}
Let $\mathcal{O}\subset\mathcal{T}_n$ be an order ideal. Then, the followings hold. 
\begin{enumerate}
    \item $\forall o\in\mathcal{O}\backslash\{1\}$, $\exists k \in \{1,\ldots,n\}$, $\pdiv{o}{x_k}/\degreek{k}{o} \in \mathcal{O}$. 
    \item $\forall b \in \partial\mathcal{O}$, $\exists k \in \{1,\ldots,n\}$, $\pdiv{b}{x_k}/\degreek{k}{b} \in \mathcal{O}$.
\end{enumerate}
\end{lemma}
\begin{proof}
\textit{Proof of (1).}
Note that, if $o \ne 1$, there always exists some $x_k$ such that $\degreek{k}{o} > 0$ (i.e., $\pdiv{o}{x_k} \ne 0$) because the total degree of $o$ is positive. Let $o = \prod_{l=1}^n x_l^{\alpha_l}\in\mathcal{O}$, where $\alpha_l\in \mathbb{Z}_{\ge 0}$ and $\alpha_k > 0$. Then, $\pdv*{o}{x_k} = \degreek{k}{o}x_k^{\alpha_{k}-1}\prod_{l\ne k} x_l^{\alpha_l}$. Because $x_k^{\alpha_{k}-1}\prod_{l\ne k} x_l^{\alpha_l}$ divides $o$, it holds $\pdiv{o}{x_k}/\degreek{k}{o} \in \mathcal{O}$. 

\textit{Proof of (2).}
For $b \in \partial\mathcal{O}$, we can write $b = x_k o$ for some $x_k$ and $o\in\mathcal{O}$. Thus, $b = x_k \prod_{l=1}^n x_l^{\alpha_l}\in\mathcal{O}$ and $\pdv*{b}{x_k} = (\alpha_k+1)\prod_{l=1}^n x_l^{\alpha_l} = \degreek{k}{b} o$; hence, $\pdiv{b}{x_k}/\degreek{k}{b} = o \in \mathcal{O}$. 
\end{proof}

Now, we prove the validity of gradient-weighted normalization in border basis computation.

\begin{proposition}\label{prop:gcnorm-theorem}
Let $G \subset\mathcal{R}_n$ be an $\mathcal{O}$-border basis of the vanishing ideal $\mathcal{I}(X)$ of $X\subset\mathbb{R}^n$. Then, the following holds.
\begin{enumerate}
\item Any $o\in\mathcal{O}\backslash \{1\}$ has a nonzero gradient-weighted norm, i.e., $\gwnorm{o} \ne 0$. 
\item Any border term $b\in\mathcal{\partial O}$ has a nonzero gradient-weighted norm, i.e., $\gwnorm{b} \ne 0$. 
\item Any $g\in G$ has a nonzero gradient-weighted norm, i.e., $\gwnorm{g} \ne 0$. 
\end{enumerate}
\end{proposition}
\begin{proof}
\textit{Proof of (1) and (2).}
From Lemma~\ref{lemma:odiv-is-o}, for any non-constant order term and border term, say $t \in \mathcal{O}\cup\partial\mathcal{O}\backslash\{1\}$, there exists a partial derivative with a term that is again an order term. Because order terms are nonvanishing for $X$ (cf. Remark~\ref{rem:O-basis-property}), the gradient-weighted norm of $t$ is nonzero. 

\textit{Proof of (3).} From Definition~\ref{def:border-basis}, the support of a border basis polynomial $g\in G$ is $\{b\} \cup \mathcal{O}$, where $b$ is a border term. As points (1) and (2), non-constant order terms and border terms have nonzero gradient-weighted norm (equivalently, nonzero gradient norm), and the gradient-weighted norm of $g$ is nonzero. 
\end{proof}
Proposition~\ref{prop:gcnorm-theorem} indicates that gradient-weighted normalization is always valid in the basis computation. 
Furthermore, gradient-weighted unitary polynomials have a bounded gradient norm.

\begin{proposition}\label{prop:gradient-norm-ub}
    For any gradient-weighted unitary polynomial $g\in\mathcal{R}_n$ for $X\subset\mathbb{R}^n$, it is $\norm{\nabla g(X)} \le \degree{g}\sqrt{|X|}$.
\end{proposition}
\begin{proof}
    Let $g=\sum_{i=1}^s c_it_i, (c_i\in\mathbb{R}, t_i\in\mathcal{T}_n)$. In addition, we define an index set $\mathfrak{j}(g) = \{i \in \{1,\ldots, s\} \mid \gwnorm{t_i}\ne 0\}$. Then, 
    \begin{align}
        \nabla g(\mathbf{x}) 
        &= \sum_{i\in \mathfrak{j}(g)}c_i \gwnorm{t_i} \frac{\nabla t_i(\mathbf{x})}{\gwnorm{t_i}}, \\
        &= \sum_{i\in \mathfrak{j}(g)}c_i \gwnorm{t_i} \frac{\nabla t_i(\mathbf{x})}{\norm{\nabla t_i(X)}}\sqrt{\sum_{k=1}^n {\degreek{k}{t_i}^2}},
    \end{align}
    Using $\sqrt{\sum_k {\degreek{k}{t_i}^2}}\le \degree{t_i} \le \degree{g}$, the triangle inequality, and $\gwnorm{g} = \sqrt{\sum_{i=1}^s c_i^2 \gwnorm{t}^2} = 1$, 
    \begin{align}
        \|\nabla g(\mathbf{x})\|
        &\le \degree{g} \sum_{i\in \mathfrak{j}(g)}|c_i| \gwnorm{t_i} \frac{\norm{\nabla t_i(\mathbf{x})}}{\norm{\nabla t_i(X)}}, \\
        &\le \degree{g} \sqrt{\sum_{i\in \mathfrak{j}(g)}c_i^2 \gwnorm{t_i}^2} \sqrt{\sum_{i\in \mathfrak{j}(g)}\frac{\norm{\nabla t_i(\mathbf{x})}^2}{\norm{\nabla t_i(X)}^2}},\\
        &= \degree{g}\sqrt{\sum_{i\in \mathfrak{j}(g)}\frac{\norm{\nabla t_i(\mathbf{x})}^2}{\norm{\nabla t_i(X)}^2}},
    \end{align}
    where at the second inequality, we employed the Cauchy--Schwarz inequality. Thus,
    \begin{align}
        \norm{\nabla g(X)} 
        &= \sqrt{\sum_{\mathbf{x}\in X} \norm{\nabla g(\mathbf{x})}^2},\\
        &\le \degree{g}\sqrt{\sum_{\mathbf{x}\in X}\sum_{i\in \mathfrak{j}(g)}\frac{\norm{\nabla t_i(\mathbf{x})}^2}{\norm{\nabla t_i(X)}^2}}, \\
        &= \degree{g} \sqrt{\sum_{i\in \mathfrak{j}(g)} 1}, \\
        &\le \degree{g}\sqrt{|X|}.
    \end{align}
    At the last inequality, we used $\abs{\mathrm{supp}(g)\backslash \{1\}}\le\abs{\mathcal{O}}\le \abs{X}$.
\end{proof}

\begin{remark}
If we use $Z = 1$ in Definition~\ref{def:gradient-norm}, the inequality in Proposition~\ref{prop:gradient-norm-ub} $\norm{\nabla g(X)} \le \sqrt{|X|}$, which makes the bound degree-independent.
\end{remark}

Proposition~\ref{prop:gradient-norm-ub} implies that, for small perturbation $\mathbf{p}$ on $\mathbf{x}$, the two evaluation values $g(\mathbf{x})$ and $g(\mathbf{x}+\mathbf{p})$ are close to each other and the difference can be bounded by the constant scaling of the magnitude of the perturbation. Later, this will be confirmed by Proposition~\ref{prop:perturbation}. However, this is not the case with coefficient normalization because the (coefficient-)unitary polynomial does not necessarily indicate a small gradient (cf.~Example~\ref{example:large-gcnorm}). Therefore, an approximately vanishing polynomial $g$ for a perturbed point $\mathbf{x}$ can be overfitting to it, and $g$ may not be well approximately vanishing for the unperturbed point $\mathbf{x}^{*} = \mathbf{x} - \mathbf{p}$, where $\mathbf{p} := \mathbf{x} - \mathbf{x}^{*}$. 

\section{Approximate computation of border bases with gradient-weighted normalization}
We will now present a method to introduce gradient-weighted normalization into the existing approximate border basis constructions (particularly, the approximate vanishing ideal~(AVI)-family methods). Almost all the AVI-family methods rely on solving eigenvalue problems (or SVD). 
Gradient-weighted normalization can be introduced by simply replacing these problems by generalized eigenvalue problems. Other methods that do not rely on eigenvalue problems solve simple quadratic programs (e.g., least-squares problems). Therefore, we consider that the proposed method can be integrated with these methods as well, owing to its simplicity.

To avoid an unnecessary abstract discussion, we adopted the ABM algorithm~\cite{limbeck2013computation} as an example because it is simple and offers various advantages over the AVI algorithm. 

\subsection{The ABM algorithm with gradient-weighted normalization}\label{sec:AMB-with-gwn}
\begin{algorithm}[t]
\DontPrintSemicolon
  \KwInput{$X, \epsilon, \sigma$}
  \KwOutput{$G, \mathcal{O}$}
  $G = \{\}, \mathcal{O} = \{1\}$\\
  
  \For{$d=1,2,\ldots $}{
  $L= \qty{ b\in \mathcal{\partial\mathcal{O}} \mid \degree{b} = d}$
  \tcp*{Assuming the terms are in the ascending order w.r.t. $\sigma$}
  \If{$|L| = 0$}{
  Return $G, \mathcal{O}$ and terminate.
  } 
  \For{$b$ in $L$}{
  Solve the generalized eigenvalue problem~Eq.~\eqref{eq:gep} and obtain $(\lambda_{\min}, \mathbf{v}_{\min})$.
      \If{$\sqrt{\lambda} \le \epsilon$}{
      \tcc{$\mathcal{O}=\{o_1,o_2,\cdots, o_s\}$ and  $\mathbf{v}_{\min}=(v_1,\ldots,v_{s+1})^{\top}$}
        $g := v_1b + v_2o_1 + \cdots +v_{s+1}o_{s}$\\
        $G = G \cup \{g\}$
      }
      \Else{
        $\mathcal{O} = \mathcal{O}\cup \{b\}$
      }
  }
 }

\caption{\\ The ABM algorithm with gradient-weighted normalization}\label{alg:ABMGN}
\end{algorithm}
Given a finite set of points $X\subset\mathbb{R}^n$, an error tolerance $\epsilon\ge 0$, and a degree-compatible term ordering $\sigma$, the ABM algorithm collects the order terms and approximately vanishing polynomials from lower to higher degrees. 
At degree 0, $\mathcal{O} = \{1\}$ and $G=\{\}$ are prepared. At degree $d\ge 1$, the degree-$d$ terms are prepared as $L = \{b \in \partial\mathcal{O} \mid \degree{b} = d\}$. If $L$ is empty, the algorithm outputs $(\mathcal{O},G)$ and terminates; otherwise, the following steps \textbf{S1--S3} are repeated until $L$ becomes empty. 
\begin{itemize}
    \item[\textbf{S1}] Select the smallest\footnote{In the original paper~\cite{limbeck2013computation}, the largest term is selected. We consider the smallest term should be selected first because the term $b$ is a potential leading term (or border term), and thus must always be larger than the terms in the tentative $\mathcal{O}$.} $b \in L$ with respect to $\sigma$ and remove $b$ from $L$. 
    
    \item[\textbf{S2}] Let $M, D$ be $M = \mqty(b(X) & O(X))$ and 
    \begin{align}
        D &= \mathrm{diag}\qty(\gwnorm{b}, \gwnorm{o_1},\ldots, \gwnorm{o_s}),
    \end{align}
    where $\mathcal{O} = \{o_1,\ldots, o_s\}$, and $\mathrm{diag}(\cdots)$ denotes a diagonal matrix with the given entries in its diagonal.
    Solve the following generalized eigenvalue problem:
    \begin{align}\label{eq:gep}
        M^{\top}M\mathbf{v}_{\min} &= \lambda_{\min}D^2\mathbf{v}_{\min},
    \end{align}
     where $\lambda_{\min}$ and $\mathbf{v}_{\min}$ are the smallest generalized eigenvalue and the corresponding generalized eigenvector, respectively.
     
    \item[\textbf{S3}] If $\sqrt{\lambda_{\min}} \le \epsilon$, 
    Define a new polynomial, 
    \begin{align}\label{eq:vanishing-polynomial}
        g = v_1b + v_2o_1 + v_3o_2 + \cdots + v_{s+1}o_s,
    \end{align}
    where $\mathbf{v}_{\min} = (v_1,\ldots,v_{s+1})^{\top}$ and $G$ is updated by $G = G \cup \{g\}$. Otherwise,
    update $\mathcal{O}$ by $\mathcal{O} = \mathcal{O} \cup \{b\}$. 
\end{itemize}
Once $L$ becomes empty, we proceed to the next degree $d+1$ and construct a new $L$. If the new $L$ is empty, the algorithm outputs $(\mathcal{O},G)$ and terminates.  

\begin{remark}\label{rem:diff-from-original-ABM}
The only difference from the original ABM algorithm is that a generalized eigenvalue problem of $(M^{\top}M, D^2)$ instead of the SVD of $M$ (equivalently, an eigenvalue problem of $M^{\top}M$) is solved. If $D$ is set to an identity matrix, the algorithm is reduced to the original one. Because of this minor difference, most of the analysis (including the termination) on the original ABM algorithm remains valid. The order of magnitude of time complexity of algorithms does not change either. 
\end{remark}

\begin{proposition}\label{prop:abmgn}
The following always holds true during the process of the ABM algorithm with gradient-weighted normalization.
\begin{enumerate}
    \item Any gradient-weighted unitary polynomial $h \in \Span{\mathcal{O}}$ is not $\epsilon$-approximately vanishing for $X$.
    \item In Eq.~\eqref{eq:vanishing-polynomial}, $g$ is gradient-weighted unitary with a nonzero coefficient on $b$ and $\sqrt{\lambda_{\min}}$-approximately vanishing for $X$. 
\end{enumerate}
\end{proposition}
\begin{proof}
\textit{Proof of (1).}
A gradient-weighted unitary polynomial in $\Span{\mathcal{O}}$ with minimal extent of vanishing can be obtained by solving a generalized eigenvalue problem $(O(X)^{\top}O(X), D^{2})$, where $D$ is a diagonal matrix with the gradient-weighted norm of the terms in $\mathcal{O}$ as the diagonal entries. However, by construction, the square root of its smallest generalized eigenvalue $\lambda_{\min}$ is larger than $\epsilon$. More specifically, during the process of the algorithm, with a tentative border term $b$ and tentative order ideal $\widetilde{\mathcal{O}}$ such that $\mathcal{O}=\{b\} \cup \widetilde{\mathcal{O}}$, the generalized eigenvalue problem of $(O(X)^{\top}O(X), D^{2})$ was already solved using Eq.~(\ref{eq:gep}). Subsequently, $\widetilde{\mathcal{O}}$ was extended to $\mathcal{O}$ because of $\sqrt{\lambda_{\min}} > \epsilon$. 

\textit{Proof of (2).} We prove the claim by induction. At the initialization of the algorithm, the claim holds true. Assume that the claim holds till a certain point in the process of the ABM algorithm, and we have $\mathcal{O},G, b$ at \textbf{S1}. 
By solving Eq.~\eqref{eq:gep} at \textbf{S2}, we obtain the coefficient vector $\mathbf{v}_{\min}$ and $g = bv_1 + v_2o_1 + v_3o_2 + \cdots + v_{s+1}o_s$.
Note that solving Eq.~\eqref{eq:gep} minimizes $\norm{g(X)}^2=\mathbf{v}_{\min}^{\top}M^{\top}M\mathbf{v}_{\min}=\lambda_{\min}$ with the constraint $\gwnorm{g}^2=\mathbf{v}_{\min}^{\top}D^2\mathbf{v}_{\min}=1$. Thus, $g$ is gradient-weighted unitary and $\sqrt{\lambda_{\min}}$-approximately vanishing. From (1), no gradient-weighted unitary polynomial with support $\mathcal{O}$ is $\epsilon$-approximately vanishing. Thus, the leading coefficient $b$ in $g$ is nonzero. 
\end{proof}

Conceptually, gradient-weighted normalization normalizes a polynomial $h = \sum_{i}v_it_i$ as $h = \sum_{i}v_i \frac{t_i}{\|\nabla t_i(X)\|}$,
where $v_i\in\mathbb{R}, t_i\in\mathcal{T}_n$. Note that this presentation is inaccurate because $\|\nabla 1(X)\| = 0$, but it provides an intuition of gradient-weighted normalization. Namely, around any point $\mathbf{x}\in X$, each term $t_i(\mathbf{x}) / \|\nabla t_i(X)\|$ behaves like a linear function because the ``degree'' of the denominator is roughly $\degree{t} - 1$ (more intuitively, if $t_i$ is a univariate, $\mathrm{d}t_i(\mathbf{x})/\mathrm{d}x$ is linear). Therefore, the polynomial $h$ behaves like a linear function around the points of $X$. This intuition motivated the several important analyses in this study, including Proposition~\ref{prop:perturbation} and Theorem~\ref{thm:scaling-consistency}.

The following theorem argues that the ABM algorithm with gradient-weighted normalization can benefit from the properties that are nearly identical to those of the original ABM algorithm~(Theorem~4.3.1 in~\cite{limbeck2013computation}). 

\begin{theorem}\label{thm:abmgn}
Given a finite set of points $X\subset\mathbb{R}^n$, $\epsilon\ge 0$, and a degree-compatible term ordering $\sigma$, the ABM algorithm with gradient-weighted normalization~(Algorithm~\ref{alg:ABMGN}) computes $G\subset\mathcal{R}_n$ and $\mathcal{O}\subset\mathcal{T}_n$, which have the following properties:
\begin{enumerate}
\item All the polynomials in $G$ are gradient-weighted unitary and $G$ generates an $\epsilon$-approximately vanishing ideal of $X$. 
\item No gradient-weighted unitary polynomial that vanishes $\epsilon$-approximately on $X$ exists in $\Span{\mathcal{O}}$. 
\item If $\mathcal{O}$ is an order ideal of terms, then the set $\widetilde{G} = \{1/\mathrm{LC}_{\sigma}(g) g \mid g \in G\}$ is an $\mathcal{O}$-border prebasis, where $\mathrm{LC}_{\sigma}(\cdot)$ is the leading coefficient of a polynomial in the ordering $\sigma$.

\item If $\epsilon = 0$, the algorithm produces the same results as the Buchberger--M\"oller algorithm for border bases with gradient-weighted normalization. 
\end{enumerate}
\end{theorem}
\begin{proof}[Proof of Theorem~\ref{thm:abmgn}]
Both (1) and (2) has been proved in Proposition~\ref{prop:abmgn}.
By Proposition~\ref{prop:abmgn}, the coefficient of the leading term (border term) of polynomials in $G$ is nonzero; thus, by construction in Algorithm~\ref{alg:ABMGN}, (3) holds~(refer to the proof of the original ABM algorithm;~\cite{limbeck2013computation})
As for (4), the proof follows from the original proof of the ABM algorithm by simply replacing the eigenvalue problem with the generalized eigenvalue problem~(cf. Remark~\ref{rem:diff-from-original-ABM}). 
\end{proof}

The main difference between this theorem and the original one with coefficient normalization is as follows: in Theorem~\ref{thm:abmgn}, (i) the unitarity of polynomials is based on the gradient-weighted norm instead of the coefficient norm; and 
(ii) it lacks the claim that $\widetilde{G}$ is a \textit{$\delta$-approximate border basis} in terms of the gradient-weighted norm and \textit{$\eta$-approximate border basis} in terms of the coefficient norm, where $\delta, \eta$ are some constants. We have eliminated the claim because the proof is lengthy for the limitation of pages. Instead, we here demonstrate that the approximation of a border basis can be discussed with the coefficient norm, even when the basis is computed with the gradient-weighted norm. This will be a key lemma to prove the claim mentioned above. 

As illustrated in Examples~\ref{example:zero-gcnorm} and~\ref{example:large-gcnorm}, the gradient-weighted and coefficient norms are not generally correlated.
However, we can show that the gradient-weighted norm can impose an upper bound on the coefficient norm in approximate border basis computation. We first prove a lemma. 

\begin{lemma}\label{lemma:gradient-norm-lb}
Let $\mathcal{O}\subset\mathcal{T}_n$ be an order ideal, which is obtained by the ABM algorithm with gradient-weighted normalization for $X\subset\mathbb{R}^n$ and $\epsilon \ge 0$. Then, for any $o\in\mathcal{O}\backslash\{1\}$, it holds $\gwnorm{o} > \epsilon^{\mathrm{deg}(o)-1}\sqrt{|X|}$.
\end{lemma}
\begin{proof}
We prove the claim by induction. At $d=1$, the claim holds because of $\|\nabla o(X)\|= \sqrt{|X|}$. Next, assume at degree $d \ge 1$, the first claim holds; that is, for any $o\in \mathcal{O}$ of degree $d$, it is $\gwnorm{o} > \epsilon^{d-1}\sqrt{|X|}$. 
Let $\mathfrak{i}(o)\subset \{1,\ldots ,n\}$ be the index set such that $x_k$ can divide $o$ for $k\in \mathfrak{i}(o)$. For each $o \in \mathcal{O}$ and $k\in \mathfrak{i}(o)$, let $o^{(k)}$ be the term such that $o = x_ko^{(k)}$. For any $o \in \mathcal{O}$ of degree $d+1$, we obtain 
\begin{align}
	\gwnorm{o}^2
    &= \frac{1}{\sum_{k=1}^n \mathrm{deg}_k(o)^2}\sum_{k\in \mathfrak{i}(o)} \norm{\frac{\partial o}{\partial x_k}(X)}^2, \\
    &= \frac{1}{\sum_{k=1}^n \mathrm{deg}_k(o)^2}\sum_{k\in \mathfrak{i}(o)} \mathrm{deg}_k(o)^2 \norm{o^{(k)}(X)}^2, \\
    &= \frac{1}{\sum_{k=1}^n \mathrm{deg}_k(o)^2}\sum_{k\in \mathfrak{i}(o)} \mathrm{deg}_k(o)^2 \frac{\norm{o^{(k)}(X)}^2}
    {\gwnorm{o^{(k)}}^2} \gwnorm{o^{(k)}}^2, \\
    &> \epsilon^{2d} |X|. 
\end{align}
For the first equality, we used  $\sum_{k=1}^n \norm{\frac{\partial o}{\partial x_k}(X)}^2 = \sum_{k\in \mathfrak{i}(o)} \norm{\frac{\partial o}{\partial x_k}(X)}^2$
because $\pdiv{o}{x_k} = 0$ for $k\notin \mathfrak{i}(o)$. For the second equality, we used $\pdiv{o}{x_k} = \mathrm{deg}_k(o)o^{(k)}$. For the last inequality, we used $\sum_{k=1}^n\degreek{k}{o}^2 = \sum_{k\in \mathfrak{i}(o)}\degreek{k}{o}^2$, $\norm{o^{(k)}(X)}/\gwnorm{o^{(k)}} > \epsilon$, and the assumption at degree $d$. Thus, for any $o\in\mathcal{O} \backslash\{1\}$, it follows that $\norm{\nabla o(X)} > \epsilon^{\mathrm{deg}(o)-1} \sqrt{|X|}$. 
\end{proof}

Now, we upper bound the coefficient norm of an approximate vanishing polynomial by its gradient-weighted norm. 

\begin{proposition}\label{prop:coef-norm-ub-by-gwnorm} %
Let $(\mathcal{O}, G)\subset \mathcal{T}_n\times \mathcal{R}_n$ be the output of the ABM algorithm for $X\subset \mathbb{R}^n$ and $\epsilon \ge 0$. For any $g\in G$, the following holds: 
\begin{align}
    \cnorm{g}
    &< \frac{\sqrt{\gwnorm{g}^2+c_0^2}}{\min\{\epsilon^{\degree{g}-1}, 1\}\sqrt{|X|}},
\end{align}
where $c_0$ is the coefficient of the constant term of $g$.
Furthermore, if $\mathbf{0}\in X$, then
\begin{align}
    \cnorm{g}
    &< \frac{\sqrt{\gwnorm{g}^2+\epsilon^2}}{\min\{\epsilon^{\degree{g}-1}, 1\}\sqrt{|X|}}.
\end{align}
\end{proposition}
\begin{proof}
    Let $g = \sum_{i=0}^{s}c_i t_i$, where $c_i \in \mathbb{R}$ and $t_0 = 1$. Let $\mathbf{c} = \qty(c_1, \ldots, c_{s})^{\top}$ and 
        $D = \mathrm{diag}(\gwnorm{t_1}, \gwnorm{t_2}, \ldots, \gwnorm{t_{s}})$ (note that $c_{0}$ and $\gwnorm{t_0}$ are excluded).
    Then, 
    \begin{align}
        \gwnorm{g}^2 
        &= \mathbf{c}^{\top}D^2\mathbf{c}, \\
        &\ge \min_{i\in\{1,\ldots, s\}} \gwnorm{t_i}^2 \norm{\mathbf{c}}^2, \\
        &= \min_{i\in\{1,\ldots, s\}} \gwnorm{t_i}^2 (\cnorm{g}^2 - c_0^2).
    \end{align}
    From Lemma~\ref{lemma:gradient-norm-lb}, we have $\gwnorm{t_i} > \epsilon^{\degree{t_i}-1}\sqrt{|X|}$ and obtain
    \begin{align}
        \cnorm{g}
        &< \frac{\sqrt{\gwnorm{g}^2+c_0^2}}{\min\{\epsilon^{\degree{g}-1}, 1\}\sqrt{|X|}}
    \end{align}
    If $\mathbf{0} \in X$, then $g$ is $\epsilon$-approximately vanishing for $\mathbf{0}$ (i.e., $|c_0| \le \epsilon$). 
\end{proof}

\subsection{Advantages of gradient-weighted normalization}
Gradient-weighted normalization has two advantages. The first one is robustness against perturbations on the input points.  
\begin{proposition}\label{prop:perturbation}
    Let $(\mathcal{O}, G)\subset \mathcal{T}_n\times \mathcal{R}_n$ be the output of the ABM algorithm for $X = \{\mathbf{x}_1,\ldots, \mathbf{x}_N\}\subset\mathbb{R}_n$ and $\epsilon \ge 0$. Let $P = \{\mathbf{p}_1,\ldots, \mathbf{p}_N\} \subset\mathbb{R}_n$ be a set of small perturbations. 
    If $g \in G$ is gradient-weighted unitary, then, 
    \begin{align}
        g(X+P) \le \epsilon + \norm{P}_{\max} \degree{g}\sqrt{|X|} + o(\norm{P}_{\max}),    
    \end{align}
    where $\|P\|_{\max} = \max_{\mathbf{p}\in P} \norm{\mathbf{p}}$, and $o(\cdot)$ is the Landau's small o. 
\end{proposition}
\begin{proof} 
Using the Taylor expansion and the triangle inequality, we get
    \begin{align}
        \|g(X+P)\|
        &\le \norm{g(X)} + \norm{P}_{\max} \norm{\nabla g(X)} + o(\norm{P}_{\max}), \\
        &\le \epsilon + \norm{P}_{\max} \degree{g}\sqrt{|X|} + o(\norm{P}_{\max}),
    \end{align}
    where, at the last inequality, we used Proposition~\ref{prop:gradient-norm-ub}.
\end{proof}

\begin{remark}
The inequality in Proposition~\ref{prop:perturbation} becomes $g(X+P) \le \epsilon + \norm{P}_{\max}\sqrt{|X|} + o(\norm{P}_{\max})$, if we use $Z=1$ in Definition~\ref{def:gradient-norm}.
\end{remark}

Another advantage of using gradient-weighted normalization is that it enables the ABM algorithm to output similar bases before and after scaling input points. 

\begin{theorem}\label{thm:scaling-consistency}
Let $X\subset\mathbb{R}^n$ be a set of points, $\epsilon > 0$, $\alpha \ne 0$, and $\sigma$ be a degree-compatible term ordering. 
Suppose that $(X, \epsilon, \sigma)$ and $(\alpha X, \abs{\alpha}\epsilon, \sigma)$ are given; the ABM algorithm with gradient-weighted normalization outputs $(\mathcal{O}, G)$ and $(\widehat{\mathcal{O}}, \widehat{G})$, respectively. Then, it holds that $\mathcal{O} = \widehat{\mathcal{O}}$. In addition, a one-to-one correspondence exists between $G$ and $\widehat{G}$. For the corresponding polynomials $g \in G$ and $\widehat{g} \in \widehat{G}$, the following holds:
\begin{align}
    \widehat{g}(\alpha X) &= \alpha g(X). 
\end{align}
Furthermore, it holds that $\mathrm{supp}(g) = \mathrm{supp}(\widehat{g})$ and the coefficients of $t \in \mathrm{supp}(g)$ in $g$ and $\widehat{g}$, say $v_t, \widehat{v}_t$, satisfy $v_t = \alpha^{\degree{t}-1} \widehat{v}_t$. 
\end{theorem}
\begin{proof}
    Let us consider two processes of the ABM algorithm; one for $(X,\epsilon, \sigma)$ and the other for $(\alpha X,\abs{\alpha}\epsilon, \sigma)$. We use the notations in Algorithm~\ref{alg:ABMGN} for the former process and add $\widehat{\cdot}$ to the notations in the latter process. 
    At the initialization, the claim holds because $\mathcal{O} = \widehat{\mathcal{O}} = \{1\}$ and $G = \widehat{G} = \{\}$.
    Assume that the claim holds true for several iterations, and now we are at \textbf{S1} with $\mathcal{O} = \widehat{\mathcal{O}}$ and $(G, \widehat{G})$ that satisfy the correspondence, and $L = \widehat{L}$. 
    Note that, for any term $t\in\mathcal{T}_n$, $t(\alpha X) = \alpha^{\degree{t}}t(X)$, and $\gwnormx{t}{\alpha X} = \alpha^{\degree{t}-1}\gwnorm{t}$. Thus, $t(\alpha X)/\gwnormx{t}{\alpha X} =  \alpha t(X)/\gwnormx{t}{X}$. 
    Therefore, by defining $S=\mathrm{diag}(\alpha^{\degree{b}}, \alpha^{\degree{o_1}}, \ldots, \alpha^{\degree{o_s}})$, 
    \begin{align}
        \widehat{M}^{\top}\widehat{M}\widehat{\mathbf{v}}_{\min}
        &= \widehat{\lambda}_{\min} \widehat{D}^2 \widehat{\mathbf{v}}_{\min},\\
        \iff M^{\top}MS\widehat{\mathbf{v}}_{\min}
        &= \widehat{\lambda}_{\min} \alpha^{-2} D^2S\widehat{\mathbf{v}}_{\min},
    \end{align}
    from which we obtain $\lambda_{\min} = \alpha^{-2} \widehat{\lambda}_{\min}$ and $\mathbf{v}_{\min}\propto S\widehat{\mathbf{v}}_{\min}$ at \textbf{S2}. 
    This indicates that thresholding $\sqrt{\lambda_{\min}}$ by $\epsilon$ in the first process is equivalent to thresholding $\sqrt{\widehat{\lambda}_{\min}}$ by $|\alpha|\epsilon$ in the second one. 
    Furthermore, by comparing the constraint of each generalized eigenvalue problem, $1 = \mathbf{v}_{\min}^{\top}D^2 \mathbf{v}_{\min}$ and $1 = \widehat{\mathbf{v}}_{\min}^{\top}\widehat{D}^2 \widehat{\mathbf{v}}_{\min}$, 
    we obtain $\mathbf{v}_{\min}=\alpha^{-1} S\widehat{\mathbf{v}}_{\min}$.
    Thus, at \textbf{S3}, the coefficients of the $i$-th term $t_i$ of $g$ and $\widehat{g}$ are related as $v_i 
     = \alpha^{\degree{t_i}-1}\widehat{v}_i$. 
    In summary, if $g = v_1b + v_2o_1 + \cdots, v_{s+1}o_{s}$ is $\epsilon$-approximately vanishing for $X$, then $\widehat{g} = \widehat{v}_1\widehat{b} + \widehat{v}_2\widehat{o}_1 + \cdots, \widehat{v}_{s+1}\widehat{o}_{s}$ is $\abs{\alpha}\epsilon$-approximately vanishing for $\alpha X$, and vice versa. If $b$ is appended to $\mathcal{O}$, then it should also be appended to $\widehat{\mathcal{O}}$. Thus, $\mathcal{O} = \widehat{\mathcal{O}}$; otherwise, $g$ and $\widehat{g}$ are appended to $G$ and $\widehat{G}$, respectively. Hence, $G$ and $\widehat{G}$ maintain the correspondence. 
\end{proof}

\begin{example}\label{example:scaling-consistency}
Let $X$ be a set of six perturbed sample points from a unit circle: 
\begin{align}
X = \{(0.39, 0.89), (-0.54, 0.93), (-0.94, -0.20), (-0.58,-0.91), (0.38,-081), (0.82,-0.01)\}.    
\end{align}
Given $X$, $\epsilon=0.1$, and the graded reverse lexicographic order $\sigma$, the ABM algorithm with gradient normalization computes $G$, which contains a single quadratic polynomial $g$ that is close to a unit circle, two cubic polynomials, and two degree-four polynomials. For $(\alpha X, \alpha\epsilon, \sigma)$, where $\alpha = 0.1$, the ABM algorithm with gradient normalization gives $\widehat{G}$, which has the same configuration as $G$. Let $\widehat{g} \in \widehat{G}$ be the quadratic polynomial. 
\begin{align}
    g &= 0.491x^2 - 0.0322xy + 0.350y^2 + 0.0866x - 0.0165y - 0.382,\\ 
    \widehat{g} &= 4.91x^2 - 0.322xy + 3.50y^2 + 0.0866x - 0.0165y - 0.0382. 
\end{align}
Notably, the coefficients of the constant, linear, and quadratic terms in $g$ are $0.1^{-1}$, $0.1^{0}$, and $0.1$ times of those in $\widehat{g}$, respectively. Meanwhile, the ABM algorithm with coefficient normalization outputs the basis sets in different configurations for $(X, \epsilon, \sigma)$ and $(\alpha X, \alpha\epsilon, \sigma)$. 
\end{example}

Theorem~\ref{thm:scaling-consistency} works even in the approximate case (i.e., $\epsilon>0$) and provides a theoretical justification for the scaling of the points at the preprocessing stage. That is, even if we scale a set of points before computing a basis, e.g., for the numerical stability, there is a corresponding basis for the set of points before the scaling, and this basis can be retrieved from the basis computed from the scaled points. This is not the case with coefficient normalization.

\begin{proposition}\label{prop:coefficient-normalization-failure}
Let $X \subset\mathbb{R}^n$ be a finite set of points. Let $\epsilon, \epsilon_{\alpha} > 0$ and $\alpha \ne 0$. Let $\sigma$ be a degree-compatible term ordering. Suppose the ABM algorithm with coefficient normalization receives $(X, \epsilon, \sigma), (\alpha X, \epsilon_{\alpha}, \sigma)$ and outputs $(\mathcal{O}, G), (\mathcal{O}_{\alpha}, G_{\alpha})$, respectively. If $G$ contains polynomials of different degrees, and if polynomials in $G$ are not strictly vanishing for $X$, then, there exists an $\alpha\ne 0$ such that $\mathcal{O} \ne \mathcal{O}_{\alpha}$ (consequently, $G \ne G_{\alpha}$). 
\end{proposition}

\begin{proof}
Recall Remark~\ref{rem:diff-from-original-ABM}.
The ABM algorithm with coefficient normalization is Algorithm~\ref{alg:ABMGN} with an SVD of $M$ at Eq.~(\ref{eq:gep}).  Now, we consider two processes of the ABM algorithm; one is for $X$ and the other is for $\alpha X$. 
Let $g$ be a polynomial of the lowest degree in $G$. Let $\mathcal{M}_g = \{b_g\} \cup \mathcal{O}_{g}$ be a union of the border term of $g$ and the tentative order ideal when $g$ is obtained in the process for $X$. To obtain $g$, an SVD is performed to $\mathcal{M}_g(X)$. Now, suppose that, in the process for $\alpha X$, $\epsilon_{\alpha}$ is properly set to obtain $\mathcal{O}_g$, and now $\mathcal{M}_g(\alpha X)$ is to be dealt with by an SVD. Note that if such $\epsilon_{\alpha}$ does not exist, the claim already holds true. Let $\Sigma = \text{diag}\qty(\alpha^{\degree{t_1}},\alpha^{\degree{t_2}}, \ldots, \alpha^{\degree{t_s}})$, where $\mathcal{M}_g = \{t_1, t_2, \ldots, t_s\}$.
Using $\mathcal{M}_g(\alpha X) = \mathcal{M}_g(X)\Sigma$,
\begin{align}
     \sigma_{\min}\qty(\mathcal{M}_g(\alpha X )) &\ge \sigma_{\min}\qty(\mathcal{M}_g(X))\sigma_{\min}\qty(\Sigma), \\
    &= \sigma_{\min}\qty(\mathcal{M}_g(X))\min\qty{1, |\alpha|^{\tau}}, 
\end{align}
where $\sigma_{\min}(\cdot)$ denotes the smallest singular value of matrix, and $\tau = \degree{g}$.
When $|\alpha| \le 1$, we have $\sigma_{\min}\qty(\mathcal{M}_g(\alpha X )) \ge \|g(X)\||\alpha|^{\tau}$. Thus, $\epsilon_{\alpha}$ must satisfy $\epsilon_{\alpha} \ge |\alpha|^{\tau}\|g(X)\|$.

Next, by the assumption, the ABM algorithm does not terminate at degree $\tau$. Let $\omega > \tau$ be the highest degree of order term $o\in \mathcal{O}$. Let $\mathcal{M}_{\omega}$ be a union of a tentative border term and tentative order ideal to obtain such $o$ through an SVD. Note that the smallest singular value of a matrix can be upper bounded by the minimum norm of the column vectors of the matrix. Thus, $\sigma_{\min}(\mathcal{M}_{\omega}(X))$ can be upper bounded by the norm of the evaluation vector of a certain degree-$\omega$ order term $o_{\omega}$. 
Since $\|o_{\omega}(\alpha X)\|= |\alpha|^{\omega} \|o_{\omega}(X)\|$, we have $\sigma_{\min}\qty(\mathcal{M}_{\omega}(\alpha X)) \le |\alpha|^{\omega} \|o_{\omega}(X)\|$. 
To summarize, $\epsilon_{\alpha}$ must satisfy
\begin{align}
    |\alpha|^{\tau}\|g(X)\| \le \epsilon_{\alpha} <  |\alpha|^{\omega}\|o_{\omega}(X)\|.\label{eq:eps-range}
\end{align}
In other words, if $|\alpha|^{\omega}\|o_{\omega}(X)\| \le |\alpha|^{\tau}\|g(X)\|$, i.e., if
\begin{align}
    |\alpha| &\le \qty(\frac{\|g(X)\|}{\|o_{\omega}(X)\|})^{1/(\omega-\tau)}, \label{eq:alpha-ub}
\end{align}
then there exists no $\epsilon_{\alpha}$ such that $\mathcal{O}= \mathcal{O}_{\alpha}$. Note that $\|o_{\omega}(X)\| > \epsilon \ge \|g(X)\|$ and $\omega > \tau$. Thus, the upper bound is nontrivial (tighter than $|\alpha|\le 1$).  
\end{proof}

Proposition~\ref{prop:coefficient-normalization-failure} implies that approximate computation of border basis with coefficient normalization is scale-sensitive. Particularly, the upper bound in Eq.~\eqref{eq:alpha-ub} becomes tighter when the gap of the two degrees---the highest degree of order terms and the lowest degree of basis polynomials---increases. 
Even one recovers from $X$ a good approximate border basis that is close to the true one, this might not be the case with $\alpha X$, and vice versa (and such $\alpha \ne 0$ always exists!).
Thus, the scaling parameter $\alpha$, as well as $\epsilon$, has to be carefully selected. 
The importance of the choice of scaling parameter from an algebraic perspective has not been discussed in literature. In \cite{heldt2009approximate}, the scaling is discussed in the numerical experiments from a numerical perspective such as computational time and maximum mean extent of vanishing, concluding that scaling points to $[-1, 1]^n$ is good for the numerical quality of computation. The advantage of gradient-weighted normalization is that one can avoid such dependency on scaling, and a set of points can be arbitrarily scaled for numerically stable computation. 

\section{Numerical experiments}\label{sec:experiment}
Here, we demonstrate that the scaling of data points is a crucial factor in the success of the approximate computation of border bases. First, using three datasets, we tested the ABM algorithm's ability to retrieve the target sets of polynomials from the scaled perturbed points. We observed that coefficient normalization is sensitive to scaling, whereas the gradient-weighted normalization is not owing to its scaling consistency (Theorem~\ref{thm:scaling-consistency}). Second, through a small numerical experiment, we show that the valid range of the scaling parameter follows Proposition~\ref{prop:coefficient-normalization-failure}. 

\subsection{Configuration retrieval test}
In the approximate setting, the target polynomial system and a system calculated from the perturbed points cannot be compared directly because the number of polynomials in the two systems may be different. We performed the following simple test. 
\begin{definition}[configuration retrieval test]
Let $G \subset\mathbb{R}[x_1,\ldots, x_n]$ be a finite set of polynomials and let $T$ be the maximum degree of polynomials in $G$. An algorithm $\mathcal{A}$, which calculates a set of polynomials $\widehat{G}\subset \mathbb{R}[x_1,\ldots, x_n]$ from a set of points $X\subset\mathbb{R}^n$, is considered to successfully retrieve the configuration of $G$ if $\forall t = 0, 1, \ldots, T, |G_t| = |\widehat{G}_t|$ is satisfied, where $G_t$ and $\widehat{G}_t$ denote the sets of degree-$t$ polynomials in $G$ and $\widehat{G}$, respectively.  
\end{definition}

The configuration retrieval test verifies if the algorithm outputs a set of polynomials that has the same configuration as the target system up to the maximum degree of polynomials in the target system. This is a necessary condition for a good approximate basis construction. Furthermore, to circumvent choosing $\epsilon$ of the ABM algorithm, we performed a linear search. Thus, a run of the ABM algorithm is considered to have passed the configuration retrieval test if there exists a proper $\epsilon$.
We considered three affine varieties.
\begin{align}
    V_1 &= \qty{(x,y) \in \mathbb{R}^2 \mid (x^2+y^2)^3 - 4x^2y^2 = 0}, \\
    V_2 &= \qty{(x, y, z) \in \mathbb{R}^3 \mid x + y - z = 0, x^3 - 9(x^2 - 3y^2) = 0}, \\
    V_3 &= \qty{(x, y, z) \in \mathbb{R}^3 \mid x^2-y^2z^2+z^3 = 0}.
\end{align}
We calculated the Gr\"{o}bner and border bases (by the ABM algorithm with two normalization methods) and confirmed that for each dataset, all the bases had the same configuration. 
Using its parametric representation, fifty points (say, $X_i^{*}$) were sampled from each $V_i$.
Each $X_i^{*}$ was preprocessed by subtracting the mean and scaling to make average $L_2$ norm of points unit. The sampled points were then perturbed by an additive Gaussian noise $\mathcal{N}(\mathbf{0}, \nu I)$, where $I$ denotes the identity matrix and $\nu\in \{0.01, 0.05\}$, and then recentered again. The set of such perturbed points from $X_i^{*}$ is denoted by $X_i$. Five scales $\alpha X_i, (\alpha=0.01, 0.1, 1.0, 10, 100)$ were considered. The linear search of $\epsilon$ was conducted with $[10^{-5}\alpha, \alpha)$ with a step size $10^{-3}\alpha$. We conducted 20 independent runs for each setting, changing the perturbation to $X_i^*$. 

Tables~\ref{table:SRT-noise01} and~\ref{table:SRT-noise05} summarize the results, corresponding to the perturbation level $\nu = 0.01, 0.05$, respectively. We first focus on Tables~\ref{table:SRT-noise01}. The success rate column displays the ratio of the successful configuration retrieval (i.e., the existence of a valid range of $\epsilon$) to 20 runs. With gradient-weighted normalization, the ABM algorithm succeeded in all datasets and scales (except $(V_2, \alpha=0.01)$), whereas with coefficient normalization, it succeeded only in specific scales (not necessarily $\alpha=1$). For numerically stable computation, the data points must be preprocessed in a certain range (in our case, mean-zero, unit average $L_2$ norm, and $\alpha=1$). However, our experiment shows that, with coefficient normalization, such preprocessing can lead the approximate border basis construction to fail. In contrast, gradient-weighted normalization provides robustness against such preprocessing. 

Another observation is that the valid range of $\epsilon$ and the extent of vanishing of gradient-weighted normalization both change in proportion to the scale. For example, at $V_2$, the ranges of valid $\epsilon$ changes as $(1.94, 2.17)\times 10^{-n}, (n=-3, -2, \ldots, 1)$ and the extent of vanishing changes as $1.80\times 10^{-n}, (n=-3, -2, \ldots, 1)$.  This tendency is supported by the scaling consistency. From Theorem~\ref{thm:scaling-consistency}, if the configuration retrieval test is passed with $(\alpha X, |\alpha|\epsilon)$ for some nonzero $\alpha =\alpha_0$, then it can be also passed by any other $\alpha \ne 0$. Note that although the extent of vanishing appear to be large at $\alpha=100$, the signal to noise ratio remains unchanged.
With coefficient normalization, the range of $\epsilon$ and the extent of vanishing change in an inconsistent way. At $V_3$, between $\alpha=0.01$ and $\alpha=1.0$, the scale of the range of $\epsilon$ differs by three order, while between $\alpha=1.0$ and $\alpha=100$ share the same order. For $\alpha=0.1$, the successful case was only $\epsilon=10^{-5}$, where the initial value and the step size of the linear search are $10^{-6}$ and $10^{-4}$, respectively.

When the the perturbation level increases to $\nu=0.05$~(Table~\ref{table:SRT-noise05}), similar results were observed; gradient-weighted normalization showed more robust against scaling of points than coefficient normalization, and the range of $\epsilon$ and the extent of vanishing changed proportional to the scaling. 
Besides, coefficient normalization resulted in lower success rate at $V_2$ and $V_3$. For example, at $(V_2, \alpha=0.1)$, the success rate dropped from 1.00 to 0.00. In contrast, gradient-weighted normalization retained its performance. This result implies the better stability of gradient-weighted normalization against perturbation.

\begin{table*}
    \centering
    \caption{Summary of the configuration retrieval test of 20 independent runs with 1 \% noise. Columns \textit{coeff.} and \textit{grad. w.} present the coefficient and gradient-weighted normalization, respectively. Column \textit{coeff. dist.} presents the distance of the normalized coefficient vectors between the systems, from 500 unperturbed and 100 perturbed points, repectively. Column \textit{e.v.} denotes the extent of vanishing at the unperturbed points. The values of the range, coefficient distance, and extent of vanishing are averaged values over 20 independent runs. As indicated by the success rate, the proposed gradient-weighted normalization approach is robust and consistent (see the proportional change in the range and the extent of vanishing) to the scaling, whereas coefficient normalization is not.}\label{table:SRT-noise01}
\begin{tabular}{c|crcccc}
    \toprule
        dataset & normalization & scaling $\alpha$ & range & coeff. dist. & e.v. & success rate \\
    \midrule
    \multirow{10}{*}{$V_1$} & \multirow{5}{*}{coeff.}  & 0.01        & --                             & --         & --                          & 0.00 [00/20]       \\
                       &                              & 0.1         & --                             & --         & --                          & 0.00 [00/20]       \\
                       &                              & 1.0         & [2.28, 2.61] $\times 10^{-2}$  &  1.39      & 1.27 $\times 10^{-2}$       & 1.00 [20/20]       \\
                       &                              & 10          & [1.80, 1.91] $\times 10^{-0}$  &  1.41      & 1.28                        & 1.00 [20/20]       \\
                       &                              & 100         & [1.31, 1.41] $\times 10^{-0}$  &  1.34      & 1.26                        & 0.85 [17/20]       \\
                       & \multirow{5}{*}{grad. w.}    & 0.01        & --                             &            & --                          & 0.00 [00/20]        \\
                       &                              & 0.1         & [3.33, 4.20] $\times 10^{-3}$  &  1.32      & 4.91 $\times 10^{-3}$       & 1.00 [20/20]       \\
                       &                              & 1.0         & [3.33, 4.20] $\times 10^{-2}$  &  1.32      & 4.91 $\times 10^{-2}$       & 1.00 [20/20]       \\
                       &                              & 10          & [3.33, 4.20] $\times 10^{-1}$  &  1.41      & 4.91 $\times 10^{-1}$       & 1.00 [20/20]       \\
                       &                              & 100         & [3.33, 4.20] $\times 10^{-0}$  &  1.41      & 4.91                        & 1.00 [20/20]       \\
    \midrule
    \multirow{10}{*}{$V_2$} & \multirow{5}{*}{coeff.}  & 0.01        & --                             & --         & --                          & 0.00 [00/20]       \\
                       &                              & 0.1         & --                             & --         & --                          & 0.00 [00/20]       \\
                       &                              & 1.0         & [0.69, 1.30] $\times 10^{-1}$  & 0.0121     & 3.67 $\times 10^{-2}$       & 1.00 [20/20]       \\
                       &                              & 10          & [1.43, 1.63] $\times 10^{-0}$  & 0.577      & 1.28                        & 1.00 [20/20]       \\
                       &                              & 100         & --                             &  --        & --                          & 0.00 [00/20]       \\
                       & \multirow{5}{*}{grad. w.}    & 0.01        & [1.94, 2.17] $\times 10^{-3}$  &  0.612     & 1.80 $\times 10^{-3}$       & 1.00 [20/20]        \\
                       &                              & 0.1         & [1.94, 2.17] $\times 10^{-2}$  &  0.518     & 1.80 $\times 10^{-2}$       & 1.00 [20/20]        \\
                       &                              & 1.0         & [1.94, 2.17] $\times 10^{-1}$  &  0.104     & 1.80 $\times 10^{-1}$       & 1.00 [20/20]        \\
                       &                              & 10          & [1.94, 2.17] $\times 10^{-0}$  &  0.455     & 1.80                        & 1.00 [20/20]        \\
                       &                              & 100         & [1.94, 2.17] $\times 10^{+1}$  &  0.681     & 1.80 $\times 10^{+1}$       & 1.00 [20/20]        \\
    \midrule
    \multirow{10}{*}{$V_3$} & \multirow{5}{*}{coeff.}  & 0.01        & --                             &  --        & --                          & 0.00 [00/20]       \\
                       &                              & 0.1         & [1.00, 1.00] $\times 10^{-6}$  &  0.836     & 1.93 $\times 10^{-6}$       & 0.65 [13/20]       \\
                       &                              & 1.0         & [7.06, 8.85] $\times 10^{-3}$  &   1.23     & 2.35 $\times 10^{-2}$       & 0.95 [19/20]       \\
                       &                              & 10          & [1.12, 1.26] $\times 10^{-0}$  &   1.20     & 2.92                        & 1.00 [20/20]       \\
                       &                              & 100         & [1.68, 1.72] $\times 10^{-0}$  &   1.05     & 2.18                        & 0.75 [15/20]       \\
                       & \multirow{5}{*}{grad. w.}    & 0.01        & [1.67, 2.36] $\times 10^{-4}$  &   1.30     & 5.97 $\times 10^{-4}$       & 1.00 [20/20]        \\
                       &                              & 0.1         & [1.67, 2.36] $\times 10^{-3}$  &   1.29     & 5.97 $\times 10^{-3}$       & 1.00 [20/20]        \\
                       &                              & 1.0         & [1.67, 2.36] $\times 10^{-2}$  &   1.29     & 5.97 $\times 10^{-2}$       & 1.00 [20/20]        \\
                       &                              & 10          & [1.67, 2.36] $\times 10^{-1}$  &   1.15     & 5.97 $\times 10^{-1}$       & 1.00 [20/20]        \\
                       &                              & 100         & [1.67, 2.36] $\times 10^{-0}$  &   1.36     & 5.97                        & 1.00 [20/20]        \\
    \bottomrule
\end{tabular}
\end{table*}

\begin{table*}
    \centering
    \caption{Summary of the configuration retrieval test of 20 independent runs with 5\% noise. Compared to the results in Table~\ref{table:SRT-noise01}, coefficient normalization decreases the success rate at $V_2$ and $V_3$, whereas gradient-weighted normalization retains the performance. }\label{table:SRT-noise05}
    \begin{tabular}{c|crcccc}
        \toprule
            dataset & normalization & scaling $\alpha$ & range & coeff. dist. & e.v. & success rate \\
        \midrule
        \multirow{10}{*}{$V_1$} & \multirow{5}{*}{coeff.}  & 0.01        & --                             & --         & --                          & 0.00 [00/20]       \\
                           &                              & 0.1         & --                             & --         & --                          & 0.00 [00/20]       \\
                           &                              & 1.0         & [2.28, 2.61] $\times 10^{-2}$  &  1.31      & 1.14 $\times 10^{-1}$       & 1.00 [20/20]       \\
                           &                              & 10          & [1.80, 1.91] $\times 10^{-0}$  &  1.41      & 4.13                        & 1.00 [20/20]       \\
                           &                              & 100         & [1.86, 1.91] $\times 10^{-0}$  &  1.20      & 3.34                        & 0.85 [17/20]       \\
                           & \multirow{5}{*}{grad. w.}    & 0.01        & --                             &            & --                          & 0.00 [00/20]        \\
                           &                              & 0.1         & [5.07, 5.66] $\times 10^{-3}$  &  1.22      & 2.61 $\times 10^{-2}$       & 1.00 [20/20]       \\
                           &                              & 1.0         & [5.07, 5.66] $\times 10^{-2}$  &  1.34      & 2.61 $\times 10^{-1}$       & 1.00 [20/20]       \\
                           &                              & 10          & [5.07, 5.66] $\times 10^{-1}$  &  1.37      & 2.61                        & 1.00 [20/20]       \\
                           &                              & 100         & [5.07, 5.66] $\times 10^{-0}$  &  1.40      & 2.61 $\times 10^{+1}$       & 1.00 [20/20]       \\
        \midrule
        \multirow{10}{*}{$V_2$} & \multirow{5}{*}{coeff.}  & 0.01        & --                             & --         & --                          & 0.00 [00/20]       \\
                           &                              & 0.1         & --                             & --         & --                          & 0.00 [00/20]       \\
                           &                              & 1.0         & --                             & --         & --                          & 0.00 [00/20]       \\
                           &                              & 10          & [3.64, 4.21] $\times 10^{-0}$  & 0.693      & 3.41                        & 1.00 [20/20]       \\
                           &                              & 100         & --                             &  --        & --                          & 0.00 [00/20]       \\
                           & \multirow{5}{*}{grad. w.}    & 0.01        & [4.11, 5.42] $\times 10^{-3}$  &  0.692     & 2.16 $\times 10^{-3}$       & 1.00 [20/20]        \\
                           &                              & 0.1         & [4.11, 5.42] $\times 10^{-2}$  &  0.667     & 2.16 $\times 10^{-2}$       & 1.00 [20/20]        \\
                           &                              & 1.0         & [4.11, 5.42] $\times 10^{-1}$  &  0.515     & 2.16 $\times 10^{-1}$       & 1.00 [20/20]        \\
                           &                              & 10          & [4.11, 5.42] $\times 10^{-0}$  &  0.565     & 2.16                        & 1.00 [20/20]        \\
                           &                              & 100         & [4.11, 5.42] $\times 10^{+1}$  &  0.692     & 2.16 $\times 10^{+1}$       & 1.00 [20/20]        \\
        \midrule
        \multirow{10}{*}{$V_3$} & \multirow{5}{*}{coeff.}  & 0.01        & --                             &  --        & --                          & 0.00 [00/20]       \\
                           &                              & 0.1         & [1.00, 1.00] $\times 10^{-6}$  &  0.416     & 2.10 $\times 10^{-6}$       & 0.30 [06/20]       \\
                           &                              & 1.0         & [0.90, 1.04] $\times 10^{-2}$  &   1.29     & 1.40                        & 0.95 [19/20]       \\
                           &                              & 10          & [1.32, 1.39] $\times 10^{-0}$  &   4.86     & 1.33                        & 1.00 [20/20]       \\
                           &                              & 100         & [1.67, 1.67] $\times 10^{-0}$  &   2.51     & 8.48 $\times 10^{-1}$       & 0.60 [12/20]       \\
                           & \multirow{5}{*}{grad. w.}    & 0.01        & [3.13, 3.68] $\times 10^{-4}$  &   1.38     & 2.37 $\times 10^{-3}$       & 1.00 [20/20]        \\
                           &                              & 0.1         & [3.13, 3.68] $\times 10^{-3}$  &   1.38     & 2.37 $\times 10^{-2}$       & 1.00 [20/20]        \\
                           &                              & 1.0         & [3.13, 3.68] $\times 10^{-2}$  &   1.34     & 2.37 $\times 10^{-1}$       & 1.00 [20/20]        \\
                           &                              & 10          & [3.13, 3.68] $\times 10^{-1}$  &   1.17     & 2.37                        & 1.00 [20/20]        \\
                           &                              & 100         & [3.13, 3.68] $\times 10^{-0}$  &   1.37     & 2.37 $\times 10^{+1}$       & 1.00 [20/20]        \\
        \bottomrule
    \end{tabular}
\end{table*}

\subsection{Valid range of scaling parameter at coefficient normalization}
Let $(G, \mathcal{O})$ and $(G_{\alpha}, \mathcal{O}_{\alpha})$ be the outputs of the ABM algorithm with coefficient normalization given $(X, \epsilon)$ and $(X_{\alpha}, \epsilon_{\alpha})$. We are interested in the range of $\alpha$ a valid range of $\epsilon_{\alpha}$, which yields $\mathcal{O} = \mathcal{O}_{\alpha}$.
Equation~\eqref{eq:alpha-ub} in Proposition~\ref{prop:coefficient-normalization-failure} argues that the valid range of $\epsilon_{\alpha}$ is lower bounded by a certain quantity (say, $\xi$). Here, we test how tight this bound is as well as confirm the scale-sensitivity of coefficient normalization again. We used $V_2$ from the previous experiment, and points are sampled, prepossessed, and perturbed by 1\% noise to obtain $X_2$ in the same way. From Table~\ref{table:SRT-noise01}, the ABM algorithm with coefficient normalization succeeds with $\alpha = 1.0, 10$.
\footnote{Note that, strictly speaking, we are not solving the same problem. In the previous experiment, a \textit{success} means that similar bases are obtained between unperturbed and perturbed points, whereas here, it means that the same order ideal (up to a certain degree) is obtained between nonscaled and scaled points.} 
Thus, we now work with $\widehat{X}_2 = \alpha X_2, (\alpha=1.0, 5.0, 10)$. We then calculate the range of $\beta$ such that given $\beta\widehat{X}_2$, the algorithm outputs the same order ideal as the one from $\widehat{X}_2$ up to a certain degree as in the configuration retrieval test. To obtain the range, linear search was performed with $\beta \in (0, 1.0]$ with step size 0.01. The results are shown in Fig.~\ref{fig:coeff-failure}. We can actually observe that the successful regions (shady parts) are lower bounded by $\xi$. Since coefficient normalization mainly works with $\alpha \in [1.0, 10]$, the bound $\epsilon$ has relatively large value at $\alpha = 1.0$, which indicates a larger failure region below $\alpha$. Since Eq.~\eqref{eq:alpha-ub} only provides a sufficient condition for nonexistence, the bound cannot be said very tight; however, the validity of the bound is confirmed. 

\begin{figure*}[t]
  \includegraphics[width=\linewidth]{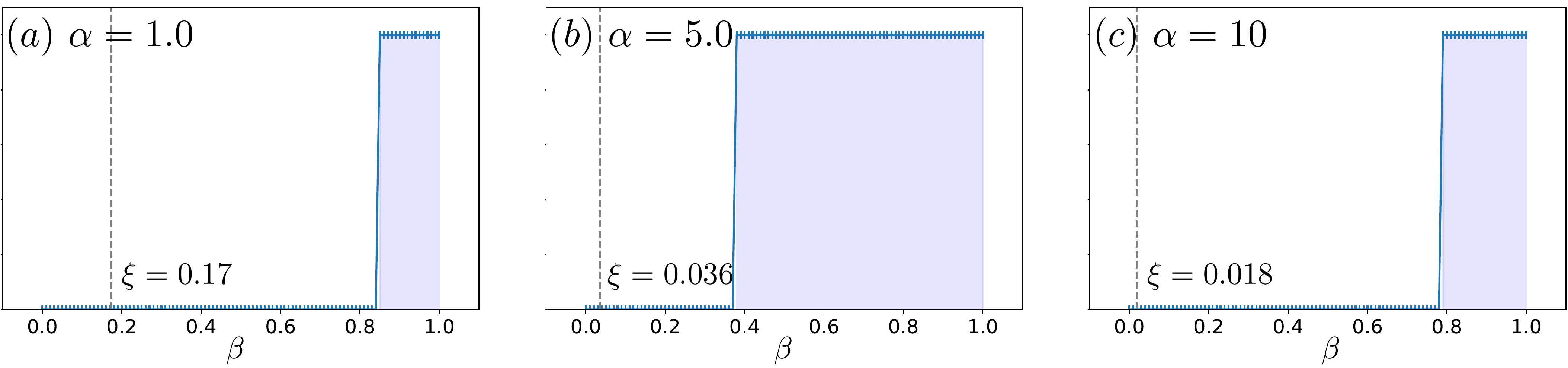}
  \caption{Theoretical lower bound $\xi$~(dashed line) for $\beta$ with which one can retrieve the same order ideal of the non-scaled case (i.e., $\beta=1$). Valid range (shaded region) of $\beta$ is found by a linear search. (a), (b), and~(c) corresponds to $\alpha=1.0, 5.0, 10$. The shaded regions are at the right side of bound as Proposition~\ref{prop:coefficient-normalization-failure} predicts.}
  
  \label{fig:coeff-failure}
\end{figure*}

\section{Conclusion}
In this study, we proposed gradient-weighted normalization for the approximate border basis computation of vanishing ideals. We showed its validity in the border basis computation by proving that the gradient-weighted norm always takes nonzero values for order terms and border basis polynomials. The introduction of gradient-weighted normalization is compatible with the existing analysis of approximate border bases and the computation algorithms. The time complexity does not change either. The data-dependent nature of gradient-weighted normalization provides several important properties (stability against perturbation and scaling consistency) to basis computation algorithms. In particular, through theory and numerical experiments, we highlighted the critical effect of the scaling of points on the success of approximate basis computation. We consider that the present study provides a new perspective and ingredients to analyze the border basis computation in the approximate setting, where perturbed points should be dealt with and stable computation is required. 

\section*{Acknowledgement}
We would like to thank Yuichi Ike for helpful discussions. This work was supported by JST, ACT-X Grant Number JPMJAX200F, Japan.

\bibliographystyle{ACM-Reference-Format}

\begin{thebibliography}{23}

\ifx \showCODEN    \undefined \def \showCODEN     #1{\unskip}     \fi
\ifx \showDOI      \undefined \def \showDOI       #1{#1}\fi
\ifx \showISBNx    \undefined \def \showISBNx     #1{\unskip}     \fi
\ifx \showISBNxiii \undefined \def \showISBNxiii  #1{\unskip}     \fi
\ifx \showISSN     \undefined \def \showISSN      #1{\unskip}     \fi
\ifx \showLCCN     \undefined \def \showLCCN      #1{\unskip}     \fi
\ifx \shownote     \undefined \def \shownote      #1{#1}          \fi
\ifx \showarticletitle \undefined \def \showarticletitle #1{#1}   \fi
\ifx \showURL      \undefined \def \showURL       {\relax}        \fi
\providecommand\bibfield[2]{#2}
\providecommand\bibinfo[2]{#2}
\providecommand\natexlab[1]{#1}
\providecommand\showeprint[2][]{arXiv:#2}

\bibitem[Abbott et~al\mbox{.}(2008)]%
        {abbott2008stable}
\bibfield{author}{\bibinfo{person}{John Abbott}, \bibinfo{person}{Claudia
  Fassino}, {and} \bibinfo{person}{Maria-Laura Torrente}.}
  \bibinfo{year}{2008}\natexlab{}.
\newblock \showarticletitle{Stable border bases for ideals of points}.
\newblock \bibinfo{journal}{\emph{Journal of Symbolic Computation}}
  \bibinfo{volume}{43} (\bibinfo{year}{2008}), \bibinfo{pages}{883--894}.
\newblock

\bibitem[Antonova et~al\mbox{.}(2020)]%
        {antonova2020analytic}
\bibfield{author}{\bibinfo{person}{Rika Antonova}, \bibinfo{person}{Maksim
  Maydanskiy}, \bibinfo{person}{Danica Kragic}, \bibinfo{person}{Sam Devlin},
  {and} \bibinfo{person}{Katja Hofmann}.} \bibinfo{year}{2020}\natexlab{}.
\newblock \showarticletitle{Analytic Manifold Learning: Unifying and Evaluating
  Representations for Continuous Control}.
\newblock \bibinfo{journal}{\emph{arXiv preprint arXiv:2006.08718}}
  (\bibinfo{year}{2020}).
\newblock

\bibitem[Fassino(2010)]%
        {fassino2010almost}
\bibfield{author}{\bibinfo{person}{Claudia Fassino}.}
  \bibinfo{year}{2010}\natexlab{}.
\newblock \showarticletitle{Almost vanishing polynomials for sets of limited
  precision points}.
\newblock \bibinfo{journal}{\emph{Journal of Symbolic Computation}}
  \bibinfo{volume}{45} (\bibinfo{year}{2010}), \bibinfo{pages}{19--37}.
\newblock

\bibitem[Fassino and Torrente(2013)]%
        {fassino2013simple}
\bibfield{author}{\bibinfo{person}{Claudia Fassino} {and}
  \bibinfo{person}{Maria-Laura Torrente}.} \bibinfo{year}{2013}\natexlab{}.
\newblock \showarticletitle{Simple varieties for limited precision points}.
\newblock \bibinfo{journal}{\emph{Theoretical Computer Science}}
  \bibinfo{volume}{479} (\bibinfo{year}{2013}), \bibinfo{pages}{174--186}.
\newblock

\bibitem[Heldt et~al\mbox{.}(2009)]%
        {heldt2009approximate}
\bibfield{author}{\bibinfo{person}{Daniel Heldt}, \bibinfo{person}{Martin
  Kreuzer}, \bibinfo{person}{Sebastian Pokutta}, {and} \bibinfo{person}{Hennie
  Poulisse}.} \bibinfo{year}{2009}\natexlab{}.
\newblock \showarticletitle{Approximate computation of zero-dimensional
  polynomial ideals}.
\newblock \bibinfo{journal}{\emph{Journal of Symbolic Computation}}
  \bibinfo{volume}{44} (\bibinfo{year}{2009}), \bibinfo{pages}{1566--1591}.
\newblock

\bibitem[Hou et~al\mbox{.}(2016)]%
        {hou2016discriminative}
\bibfield{author}{\bibinfo{person}{Chenping Hou}, \bibinfo{person}{Feiping
  Nie}, {and} \bibinfo{person}{Dacheng Tao}.} \bibinfo{year}{2016}\natexlab{}.
\newblock \showarticletitle{Discriminative Vanishing Component Analysis}. In
  \bibinfo{booktitle}{\emph{Proceedings of the Thirtieth AAAI Conference on
  Artificial Intelligence (AAAI)}}. \bibinfo{publisher}{AAAI Press},
  \bibinfo{pages}{1666--1672}.
\newblock

\bibitem[Karimov et~al\mbox{.}(2020)]%
        {karimov2020algebraic}
\bibfield{author}{\bibinfo{person}{Artur Karimov},
  \bibinfo{person}{Erivelton~G. Nepomuceno}, \bibinfo{person}{Aleksandra
  Tutueva}, {and} \bibinfo{person}{Denis Butusov}.}
  \bibinfo{year}{2020}\natexlab{}.
\newblock \showarticletitle{Algebraic Method for the Reconstruction of
  Partially Observed Nonlinear Systems Using Differential and Integral
  Embedding}.
\newblock \bibinfo{journal}{\emph{Mathematics}} \bibinfo{volume}{8},
  \bibinfo{number}{2} (\bibinfo{year}{2020}), \bibinfo{pages}{300}.
\newblock

\bibitem[Kehrein and Kreuzer(2005)]%
        {kehrein2005charactorizations}
\bibfield{author}{\bibinfo{person}{Achim Kehrein} {and} \bibinfo{person}{Martin
  Kreuzer}.} \bibinfo{year}{2005}\natexlab{}.
\newblock \showarticletitle{Characterizations of border bases}.
\newblock \bibinfo{journal}{\emph{Journal of Pure and Applied Algebra}}
  \bibinfo{volume}{196}, \bibinfo{number}{2} (\bibinfo{year}{2005}),
  \bibinfo{pages}{251--270}.
\newblock
\showISSN{0022-4049}

\bibitem[Kera and Hasegawa(2016)]%
        {kera2016noise}
\bibfield{author}{\bibinfo{person}{Hiroshi Kera} {and}
  \bibinfo{person}{Yoshihiko Hasegawa}.} \bibinfo{year}{2016}\natexlab{}.
\newblock \showarticletitle{Noise-tolerant algebraic method for reconstruction
  of nonlinear dynamical systems}.
\newblock \bibinfo{journal}{\emph{Nonlinear Dynamics}}  \bibinfo{volume}{85}
  (\bibinfo{year}{2016}), \bibinfo{pages}{675--692}.
\newblock

\bibitem[Kera and Hasegawa(2018)]%
        {kera2018approximate}
\bibfield{author}{\bibinfo{person}{Hiroshi Kera} {and}
  \bibinfo{person}{Yoshihiko Hasegawa}.} \bibinfo{year}{2018}\natexlab{}.
\newblock \showarticletitle{Approximate Vanishing Ideal via Data Knotting}. In
  \bibinfo{booktitle}{\emph{Proceedings of the Thirty-Second {AAAI} Conference
  on Artificial Intelligence (AAAI)}}. \bibinfo{publisher}{AAAI Press},
  \bibinfo{pages}{3399--3406}.
\newblock

\bibitem[Kera and Hasegawa(2019)]%
        {kera2019spurious}
\bibfield{author}{\bibinfo{person}{Hiroshi Kera} {and}
  \bibinfo{person}{Yoshihiko Hasegawa}.} \bibinfo{year}{2019}\natexlab{}.
\newblock \showarticletitle{Spurious Vanishing Problem in Approximate Vanishing
  Ideal}.
\newblock \bibinfo{journal}{\emph{IEEE Access}}  \bibinfo{volume}{7}
  (\bibinfo{year}{2019}), \bibinfo{pages}{178961--178976}.
\newblock

\bibitem[Kera and Hasegawa(2020)]%
        {kera2020gradient}
\bibfield{author}{\bibinfo{person}{Hiroshi Kera} {and}
  \bibinfo{person}{Yoshihiko Hasegawa}.} \bibinfo{year}{2020}\natexlab{}.
\newblock \showarticletitle{Gradient Boosts the Approximate Vanishing Ideal}.
  In \bibinfo{booktitle}{\emph{Proceedings of the Thirty-Fourth {AAAI}
  Conference on Artificial Intelligence (AAAI)}}. \bibinfo{publisher}{AAAI
  Press}, \bibinfo{pages}{4428--4425}.
\newblock

\bibitem[Kera and Iba(2016)]%
        {kera2016vanishing}
\bibfield{author}{\bibinfo{person}{Hiroshi Kera} {and} \bibinfo{person}{Hitoshi
  Iba}.} \bibinfo{year}{2016}\natexlab{}.
\newblock \showarticletitle{Vanishing ideal genetic programming}. In
  \bibinfo{booktitle}{\emph{Proceedings of the 2016 IEEE Congress on
  Evolutionary Computation (CEC)}}. \bibinfo{publisher}{IEEE},
  \bibinfo{pages}{5018--5025}.
\newblock

\bibitem[Kir{\'a}ly et~al\mbox{.}(2014)]%
        {kiraly2014dual}
\bibfield{author}{\bibinfo{person}{Franz~J Kir{\'a}ly}, \bibinfo{person}{Martin
  Kreuzer}, {and} \bibinfo{person}{Louis Theran}.}
  \bibinfo{year}{2014}\natexlab{}.
\newblock \showarticletitle{Dual-to-kernel learning with ideals}.
\newblock \bibinfo{journal}{\emph{arXiv preprint arXiv:1402.0099}}
  (\bibinfo{year}{2014}).
\newblock

\bibitem[Kreuzer and Robbiano(2005)]%
        {kreuzer2005computational}
\bibfield{author}{\bibinfo{person}{Martin Kreuzer} {and}
  \bibinfo{person}{Lorenzo Robbiano}.} \bibinfo{year}{2005}\natexlab{}.
\newblock \bibinfo{booktitle}{\emph{Computational commutative algebra 2}}.
  Vol.~\bibinfo{volume}{2}.
\newblock \bibinfo{publisher}{Springer Science \& Business Media}.
\newblock

\bibitem[Limbeck(2013)]%
        {limbeck2013computation}
\bibfield{author}{\bibinfo{person}{Jan Limbeck}.}
  \bibinfo{year}{2013}\natexlab{}.
\newblock \emph{\bibinfo{title}{Computation of approximate border bases and
  applications}}.
\newblock \bibinfo{thesistype}{Ph.\,D. Dissertation}. \bibinfo{school}{Passau,
  Universit{\"a}t Passau}.
\newblock

\bibitem[Livni et~al\mbox{.}(2013)]%
        {livni2013vanishing}
\bibfield{author}{\bibinfo{person}{Roi Livni}, \bibinfo{person}{David Lehavi},
  \bibinfo{person}{Sagi Schein}, \bibinfo{person}{Hila Nachliely},
  \bibinfo{person}{Shai Shalev-Shwartz}, {and} \bibinfo{person}{Amir
  Globerson}.} \bibinfo{year}{2013}\natexlab{}.
\newblock \showarticletitle{Vanishing component analysis}. In
  \bibinfo{booktitle}{\emph{Proceedings of the Thirteenth International
  Conference on Machine Learning (ICML)}}. \bibinfo{publisher}{PMLR},
  \bibinfo{pages}{597--605}.
\newblock

\bibitem[Robbiano and Abbott(2010)]%
        {robbiano2010approximate}
\bibfield{author}{\bibinfo{person}{Lorenzo Robbiano} {and}
  \bibinfo{person}{John Abbott}.} \bibinfo{year}{2010}\natexlab{}.
\newblock \bibinfo{booktitle}{\emph{Approximate Commutative Algebra}}.
\newblock \bibinfo{publisher}{Springer-Verlag Wien}.
\newblock

\bibitem[Stetter(2004)]%
        {stetter2004numerical}
\bibfield{author}{\bibinfo{person}{Hans~J. Stetter}.}
  \bibinfo{year}{2004}\natexlab{}.
\newblock \bibinfo{booktitle}{\emph{Numerical Polynomial Algebra}}.
\newblock \bibinfo{publisher}{Society for Industrial and Applied Mathematics},
  \bibinfo{address}{USA}.
\newblock
\showISBNx{0898715571}

\bibitem[Torrente(2008)]%
        {torrente2009application}
\bibfield{author}{\bibinfo{person}{Maria-Laura Torrente}.}
  \bibinfo{year}{2008}\natexlab{}.
\newblock \emph{\bibinfo{title}{Application of algebra in the oil industry}}.
\newblock \bibinfo{thesistype}{Ph.\,D. Dissertation}. \bibinfo{school}{Scuola
  Normale Superiore, Pisa}.
\newblock

\bibitem[Wang and Ohtsuki(2018)]%
        {wang2018nonlinear}
\bibfield{author}{\bibinfo{person}{Lu Wang} {and} \bibinfo{person}{Tomoaki
  Ohtsuki}.} \bibinfo{year}{2018}\natexlab{}.
\newblock \showarticletitle{Nonlinear Blind Source Separation Unifying
  Vanishing Component Analysis and Temporal Structure}.
\newblock \bibinfo{journal}{\emph{IEEE Access}}  \bibinfo{volume}{6}
  (\bibinfo{year}{2018}), \bibinfo{pages}{42837--42850}.
\newblock

\bibitem[Wang et~al\mbox{.}(2019)]%
        {wang2019polynomial}
\bibfield{author}{\bibinfo{person}{Zhichao Wang}, \bibinfo{person}{Qian Li},
  \bibinfo{person}{Gang Li}, {and} \bibinfo{person}{Guandong Xu}.}
  \bibinfo{year}{2019}\natexlab{}.
\newblock \showarticletitle{Polynomial Representation for Persistence Diagram}.
  In \bibinfo{booktitle}{\emph{2019 IEEE/CVF Conference on Computer Vision and
  Pattern Recognition (CVPR)}}. \bibinfo{pages}{6116--6125}.
\newblock

\bibitem[Wirth and Pokutta(2022)]%
        {wirth2022conditional}
\bibfield{author}{\bibinfo{person}{E. Wirth} {and} \bibinfo{person}{S.
  Pokutta}.} \bibinfo{year}{2022}\natexlab{}.
\newblock \bibinfo{title}{Conditional Gradients for the Approximately Vanishing
  Ideal}.
\newblock
\newblock
\showeprint[arxiv]{2202.03349}

\end{thebibliography}

\end{document}